\documentclass[11pt]{article}
\usepackage{fullpage}
\usepackage{times}
\usepackage{amssymb,amsmath}
\usepackage{algorithm}
\usepackage{algpseudocode}
\usepackage{todonotes}

\newtheorem{theorem}{Theorem}[section]
\newtheorem{lemma}[theorem]{Lemma}
\newtheorem{claim}[theorem]{Claim}

\newtheorem{definition}[theorem]{Definition}

\newcommand{\sq}{\hbox{\rlap{$\sqcap$}$\sqcup$}}
\newcommand{\qed}{\hspace*{\fill}\sq}
\newenvironment{proof}{\noindent {\bf Proof.}\ }{\qed\par\vskip 4mm\par}

\def\argmin{{\rm argmin}}
\def\argmax{{\rm argmax}}

\begin{document}

\date{}
\begin{titlepage}

\title{CONE-DHT: A distributed self-stabilizing algorithm for a heterogeneous storage system}

\author{Sebastian Kniesburges \\
   University of Paderborn \\
   seppel@upb.de \\
   \and
   Andreas Koutsopoulos \\
   University of Paderborn  \\
   koutsopo@mail.upb.de\\
   \and
   Christian Scheideler\\
   University of Paderborn  \\
   scheideler@mail.upb.de
   }

\maketitle \thispagestyle{empty}


\begin{abstract}
We consider the problem of managing a dynamic heterogeneous storage system in
a distributed way so that the amount of data assigned to a host in that system
is related to its capacity. Two central problems have to be solved for this:
(1) organizing the hosts in an overlay network with low degree and diameter so
that one can efficiently check the correct distribution of the data and route
between any two hosts, and (2) distributing the data among the hosts so that
the distribution respects the capacities of the hosts and can easily be
adapted as the set of hosts or their capacities change. We present distributed
protocols for these problems that are self-stabilizing and that do not need
any global knowledge about the system such as the number of nodes or the
overall capacity of the system. Prior to this work no solution was known
satisfying these properties.
\end{abstract}

\bigskip


\end{titlepage}

\section{Introduction}

In this paper we consider the problem of designing distributed protocols for a
dynamic heterogeneous storage system. Many solutions for distributed storage
systems have already been proposed in the literature. In the peer-to-peer
area, distributed hash tables (DHTs) have been the most popular choice. In a
DHT, data elements are mapped to hosts with the help of a hash function, and
the hosts are organized in an overlay network that is often of hypercubic
nature so that messages can be quickly exchanged between any two hosts. To be
able to react to dynamics in the set of hosts and their capacities, a
distributed storage system should support, on top of the usual data
operations, operations to join the system, to leave the system, and to change
the capacity of a host in the desired way. We present self-stabilizing
protocols that can handle all of these operations in an efficient way.

\subsection{Heterogeneous storage systems}

Many data management strategies have already been proposed for distributed
storage systems. If all hosts have the same capacity, then a well-known
approach called {\em consistent hashing} can be used to manage the data
\cite{conhash}. In consistent hashing, the data elements are hashed to points
in $[0,1)$, and the hosts are mapped to disjoint intervals in $[0,1)$, and a
host stores all data elements that are hashed to points in its interval. An
alternative strategy is to hash data elements and hosts to pseudo-random bit
strings and to store (indexing information about) a data element at the host
with longest prefix match \cite{Plaxton}. These strategies have been realized
in various DHTs including CAN \cite{RF+01:can}, Pastry \cite{RD01:pastry} and
Chord \cite{SM+01:chord}. However, all of these approaches assume hosts of
uniform capacity, despite the fact that in P2P systems the peers can be highly
heterogeneous.

In a heterogeneous setting, each host (or node) $u$ has its specific capacity
$c(u)$ and the goal considered in this paper is to distribute the data among
the nodes so that node $u$ stores a fraction of $\frac{c(u)}{\sum_{\forall
v}{c(v)}}$ of the data. The simplest solution would be to reduce the
heterogeneous to the homogeneous case by splitting a host of $k$ times the
base capacity (e.g., the minimum capacity of a host) into $k$ many virtual
hosts. Such a solution is not useful in general because the number of virtual
hosts would heavily depend on the capacity distribution, which can create a
large management overhead at the hosts. Nevertheless, the concept of virtual
hosts has been explored before (e.g., \cite{godfrey,rao, skewccc+}). In
\cite{godfrey} the main idea is not to place the virtual hosts belonging to a
real host randomly in the identifier space but in a restricted range to
achieve a low degree in the overlay network. However, they need an estimation
of the network size and a classification of nodes with high, average, and low
capacity. A similar approach is presented in \cite{skewccc+}.
Rao et al. \cite{rao} proposed some schemes also based on virtual servers, where the data is moved from heavy nodes to light nodes to balance the load after the data assignment, so and data movement is induced even without joining or leaving nodes.
In \cite{supernode} the authors organize the nodes into clusters, where a
super node (i.e., a node with large capacity) is supervising a cluster of
nodes with small capacities. Giakkoupis et al. \cite{Giakkoupis} present an
approach which focuses on homogeneous networks but also works for
heterogeneous one. However, updates can be costly.

Several solutions have been proposed in the literature that can manage
heterogeneous storage systems in a centralized way, i.e. they consider data
placement strategies for heterogeneous disks that are managed by a single
server ~\cite{rio,slicing,spread,diskarray,broadscale,redShare} or assume a
central server that handles the mapping of data elements to a set of hosts
~\cite{cone,brink2,brink}. We will only focus on the most relevant ones for
our approach. In \cite{brink} Brinkmann et al. introduced several criteria a
placement scheme needs to fulfill, like a faithful distribution, efficient
localization, and fast adaptation. They introduce two different data placement
strategies named SHARE and SIEVE that fulfill their criteria. To apply their
approach, the number of nodes and the overall capacity of the system must be
known. In \cite{redShare} redundancy is added to the SHARE strategy to allow a
fair and redundant data distribution, i.e. several copies of a data element
are stored such that no two copies are stored on the same host. Another
solution to handle redundancy in heterogeneous systems is proposed in
\cite{spread}, but also here the number of nodes and the overall capacity of
the system must be known. The only solution proposed so far where this is not
the case is the approach by Schindelhauer and Schomaker \cite{cone}, which we
call {\em cone hashing}. Their basic idea is to assign a distance function to
each host that scales with the capacity of the host. A data element is then
assigned to the host of minimum distance with respect to these distance
functions. We will extend their construction into a self-stabilizing DHT with
low degree and diameter that does not need any global information and that can
handle all operations in a stable system efficiently with high probability
(w.h.p.)\footnote{ I.e., a probability of $1-n^{-c}$ for any constant $c>0$}.

\subsection{Self-Stabilization}

A central aspect of our self-stabilizing DHT is a self-stabilizing overlay
network that can be used to efficiently check the correct distribution of the
data among the hosts and that also allows efficient routing. There is a large
body of literature on how to efficiently maintain overlay networks, e.g.,
\cite{AS03:skip,AS04:hyperring,BK+04:Pagoda,RD01:pastry,HJ+03:skipnet,KSW05:self-repair,MNR02:viceroy,NW07:continous-discrete,RF+01:can,
SM+01:chord}. While many results are already known on how to keep an overlay
network in a legal state, far less is known about self-stabilizing overlay
networks. A self-stabilizing overlay network is a network that can recover its
topology from an arbitrary weakly connected state. The idea of
self-stabilization in distributed computing was introduced in a classical paper
by E.W. Dijkstra in 1974 \cite{D74:self-stab} in which he looked at the
problem of self-stabilization in a token ring. In order to recover certain
network topologies from any weakly connected state, researchers have started
with simple line and ring networks (e.g.
\cite{CF05:stab-ring,SR05:ring,ORS07:lin}). Over the years more and more
network topologies were considered
\cite{DK08,JRS+09:delaunay,JR+09:skip+,DT09,DT10,smallworld}. In
\cite{KKS11} the authors present a self-stabilizing algorithm for the Chord
DHT \cite{SM+01:chord}, which solves the uniform case, but the problem of
managing heterogeneous hosts in a DHT was left open, which is addressed in
this paper. To the best of our knowledge this is the first self-stabilizing
approach for a distributed heterogeneous storage system.

In this paper we present a self-stabilizing overlay network for a distributed
heterogeneous storage system based on the data assignment presented in
\cite{cone}.

\subsection{Model}

\subsubsection{Network model}

We assume an asynchronous message passing model for the CONE-DHT which is
related to the model presented in \cite{NNS11:corona} by Nor et al. The
overlay network consists of a static set $V$ of $n$ nodes or hosts. We further
assume \emph{fixed identifiers} (ids) for each node. These identifiers are
\emph{immutable} in the computation, we only allow identifiers to be compared,
stored and sent. In our model the identifiers are used as addresses, such that
by knowing the identifier of a node another node can send messages to this
node. The identifiers form a unique order. The communication between nodes is
realized by passing messages through channels. A node $v$ can send a message
to $u$ through the channel $Ch_{v,u}$. We denote the channel $Ch_u$ as the
union of all channels $Ch_{v,u}$. We assume that the capacity of a channel is
unbounded and no messages are lost. Furthermore we assume that for a
transmission pair $(v,u)$ the messages sent by $v$ are received by $u$ in the
same order as they are sent, i.e. $Ch_{v,u}$ is a FIFO channel. Note that this
does not imply any order between messages from different sending nodes. For
the channel we assume \emph{eventual delivery} meaning that if there is a
state in the computation where there is a message in the channel $Ch_u$ there
also is a later state where the message is not in the channel, but was
received by the process. We distinguish between the \emph{node state}, that is
given by the set of identifiers stored in the internal variables $u$ can communicate with, and the \emph{channel state}, that is
given by all identifiers contained in messages in a channel $Ch_u$. We
model the network by a directed graph $G=\left(V,E\right)$. The set of edges
$E$ describes the possible communication pairs. $E$ consists of two subsets:
the \emph{explicit edges} $E_e=\left\{\left(u,v\right): v \text{ is in u's
node state}\right\}$ and the \emph{implicit edges}
$E_i=\left\{\left(u,v\right): v \text{ is in u's channel state}\right\}$, i.e.
$E= E_e \cup E_i$. Moreover we define $G_e=\left(V,E_e\right)$.

\subsubsection{Computational Model}
An action has the form $<guard>\rightarrow <command>$. \emph{guard} is a
predicate that can be true or false. \emph{command} is a sequence of
statements that may perform computations or send messages to other nodes. We
introduce one special guard predicate $\tau$ called the \emph{timer
predicate}, which is periodically true; i.e. according to an internal clock $\tau$ becomes true after a number of clock cycles and is false  the other times,  and allows the nodes to perform
periodical actions. A second predicate is true if a message is received by a
node.
The \emph{program state} is defined by the node states and the channel
states of all nodes, i.e. the assignment of values to every variable of each
node and messages to every channel. We call the combination of the node
states of all nodes the \emph{node state of the system} and the combination
of the channel states of all nodes is called the \emph{channel state of the
system}.
An action is enabled in some state if its guard is true and disabled
otherwise.
A \emph{computation} is a sequence of states such that for each state $s_i$
the next state $s_{i+1}$ is reached by executing an enabled action in $s_i$.
By this definition, actions can not overlap and are executed atomically
giving a sequential order of the executions of actions. For the execution of
actions we assume \emph{weak fairness} meaning that if an action is enabled
in all but finitely many states of the computation then this action is
executed infinitely often.

We state the following requirements on our solution: \emph{Fair load
balancing}: every node with x\% of the available capacity gets x\% of the
data. \emph{Space efficiency}: Each node stores at most\\ $\mathcal
O(|\text{data assigned to the node}| + \log n)$ information. \emph{Routing
efficiency}: There is a routing strategy that allows efficient routing in at
most $\mathcal O(\log n)$ hops. \emph{Low degree}: The degree of each node is
limited by $\mathcal O(\log n)$. Furthermore we require an algorithm that
builds the target network topology in a \emph{self-stabilizing} manner, i.e.,
any weakly connected network $G=(V,E)$ is eventually transformed into a
network so that a (specified) subset of the explicit edges forms the target
network topology ({\em convergence}) and remains stable as long as no node
joins or leaves ({\em closure}).

\subsection{Our contribution}

We present a self-stabilizing algorithm that organizes a set of heterogeneous
nodes in an overlay network such that each data element can be efficiently
assigned to the node responsible for it. We use the scheme described in
\cite{cone} (which gives us good load balancing) as our data management scheme
and present a distributed protocol for the overlay network, which is efficient
in terms of message complexity and information storage and moreover works in a
self-stabilizing manner. The overlay network efficiently supports the basic
operations of a heterogeneous storage system, such as the joining or leaving
of a node, changing the capacity of a node, as well as searching, deleting and
inserting a data element. In fact we show the following main result:

\begin{theorem}\label{theo:main}
There is a self-stabilizing algorithm for maintaining a heterogeneous storage
system that achieves fair load-balancing, space efficiency and routing
efficiency, while each node has a degree of $\mathcal O(\log n)$ w.h.p. The
data operations can be handled in $O(\log n)$ time in a stable system, and if
a node joins or leaves a stable system or changes its capacity, it takes at
most $\mathcal O(\log^2 n)$ structural changes, i.e., edges that are created
or deleted, until the system stabilizes again.
\end{theorem}

\subsection{Structure of the paper}

The paper is structured as follows: 
In Section ~\ref{cone-dht} we describe our target network and its properties.
In Section ~\ref{alg} we present our self-stabilizing protocol and prove that
it is correct. Finally, in Section ~\ref{ops} we describe the functionality of
the basic network operations.


\section{The \emph{CONE}-DHT}\label{cone-dht}

\subsection{The original CONE-Hashing}

Before we present our solution, we
first give some more details on the original CONE-Hashing \cite{cone} our
approach is based on. In \cite{cone} the authors present a centralized
solution for a heterogeneous storage system in which the nodes are of
different capacities. We denote the  capacity of a node $u$ as $c(u)$. We use
a hash function $h:V\mapsto [0,1)$ that assigns to each node a hash value. A
data element of the data set $D$ is also hashed by a hash function $g: D
\mapsto [0,1)$. W.l.o.g. we assume that all hash values and capacities are
distinct.  According to \cite{cone} each node has a capacity function
$C_u(g(x))$, which determines which data is assigned to the node. A node is
\emph{responsible} for those elements $d$ with $C_u(g(d))=\min_{v
\in V}\{C_v(g(d))\}$, i.e. $d$ is assigned to $u$. We denote by $R(u)=\{x \in
[0,1): C_u(x)=\min_{v \in V}\{C_v(x)\} \}$ the \emph{responsibility range} of
$u$  (see Figure ~\ref{example}). Note that $R(u)$ can consist of several intervals in $[0,1)$.
In the original paper \cite{cone}, the authors considered two special cases
of capacity functions, one of linear form
$C^{lin}_u(x)=\frac{1}{c(u)}|x-h(u)|$ and of logarithmic form
$C^{log}_u(x)=\frac{1}{c(u)}(-log(|1-(x-h(u))|)$. For these capacity
functions the following results were shown by the authors \cite{cone}:

\begin{theorem}\label{theo:coneOriginal1}
A data element $d$ is assigned to a node $u$ with probability
$\frac{c(u)}{\sum_{v \in V}{c(v)}-c(u)}$ for linear capacity functions
$C^{lin}_u(x)$ and with probability $\frac{c(u)}{\sum_{v\in V}{c(v)}}$ for logarithmic capacity functions $C^{log}_u(x)$. Thus in expectation fair load balancing can be achieved by using a logarithmic
capacity function $C^{log}_u(x)$.
\end{theorem}

The CONE-Hashing supports the following operations for a data element $d$ or
a node $v$:

\begin{itemize}\itemsep0.1pt
        \item \emph{Search(d)}: Returns the node $u$ such that $g(d) \in
R(u)$.
        \item \emph{Insert(d)}: $d$ is assigned to  the node returned by $Search(d)$.
        \item \emph{Delete(d)}: $d$ is removed from the node returned by $Search(d)$.
        \item \emph{Join(v)}: For all $u \in V$ the responsibility ranges
$R(u)$ are updated and data elements $d$, with $g(d)\in R(v)$ are moved
to $v$.
        \item \emph{Leave(v)}: For all $u \in V$ the responsibility ranges
$R(u)$ are updated and data elements $d$ assigned to $v$ are moved to nodes
$w$ such that $g(d)\in R(w)$.
        \item \emph{CapacityChange(v)}: For all $u \in V$ the responsibility
ranges $R(u)$ are updated and data elements $d$ not assigned to $v$, but with $g(d)\in R(v)$ are moved to $v$ while data elements $d'$ assigned to $v$
but with $g(d)\in R(w)$ are moved to nodes $w$.
\end{itemize}

Moreover, the authors showed that the fragmentation is relatively small for the  logarithmic
capacity function, with each node having in expectation a logarithmic number of intervals it is responsible for.
In the case of the linear function, it can be shown that this number is only constant in expectation.

In \cite{cone} the authors further present a data structure to efficiently support the described
operations in a centralized approach. For their data structure they showed that
there is an algorithm that determines for a data element $d$ the
corresponding node $u$ with $g(d) \in R(u)$ in expected time $\mathcal
O(\log n)$. The used data structure has a size of $\mathcal O(n)$ and the
joining, leaving and the capacity change of a node can be handled efficiently.

In the following we show that CONE-Hashing can also be realized by using a
distributed data structure. Further the following challenges have to be
solved. We need a suitable topology on the node set $V$ that supports an
efficient determination of the responsibility ranges $R(u)$ for each node
$u$ . The topology should also support an efficient \emph{Search(d)}
algorithm, i.e. for an \emph{Search(d)} query inserted at an arbitrary node
$w$, the node $v$ with $g(d) \in R(v)$ should be found. Furthermore a
\emph{Join(v)}, \emph{Leave(v)}, \emph{CapacityChange(v)} operation should
not lead to a high amount of data movements, (i.e. not more than the data
now assigned to $v$ or no longer assigned to $v$ should be moved,) or a high amount of structural changes ( i.e. changes in the topology built on
$V$). All these challenges will be solved by our CONE-DHT.

\subsection{The \emph{CONE}-DHT}\label{cone}

In order to construct a heterogeneous storage network in the distributed case, we have to deal with the challenges mentioned above.
For that, we introduce the  \emph{CONE}-graph, which is an overlay network that, as we show, can support efficiently a heterogeneous storage system.

\subsubsection{The network layer}

We define the \emph{CONE} graph as a graph $G^{CONE}=(V,E^{CONE})$, with $V$ being the hosts of our storage system.

For the determination of the edge set, we need following definitions, with respect to a node $u$:

 \begin{itemize}\itemsep0.1pt
\item $succ^+_1(u)=argmin\{h(v):h(v)>h(u) \wedge c(v)>c(u)\}$
is the next node at the right of $u$ with larger capacity, and we call it the first larger successor of $u$.
Building upon this, we define recursively the i-th larger successor of $u$ as:
$succ^+_i(u)=succ^+_1(succ^+_{i-1}(u)), \forall i>1$, and the union of all larger successors as
 $S^+(u)=\bigcup_{ i} succ^+_i(u)$.

\item The first larger predecessor of $u$ is defined as:
 $pred^+_1(u)=argmax\{h(v):h(v)<h(u) \wedge c(v)>c(u)\}$
i.e. the next node at the left of $u$ with larger capacity. The i-th larger predecessor of $u$ is:
 $pred^+_i(u)=pred^+_1(pred^+_{i-1}(u)), \forall i>1$, and the union of all larger predecessors as
 $P^+(u)=\bigcup_{ i} pred^+_i(u)$.

\item We also define the set of the smaller successors of $u$, $S^-(u)$, as the set of all nodes $v$, with $u=pred_1^+(v)$, and the set of the smaller predecessors of $u$, $P^-(u)$ as the set of all nodes $v$, such that $u=succ_1^+(v)$.

\end{itemize}

Now we can define the edge-set of a node in $G^{CONE}$.
 
\begin{definition}
$(u,v) \in E^{CONE}$ iff $v \in S^+(u) \cup P^+(u) \cup S^-(u) \cup P^-(u)$
\end{definition}

We define also the neighborhood set of $u$ as $N_u= S^+(u) \cup P^+(u) \cup S^-(u) \cup P^-(u)$.
In other words, $v$ maintains connections to each node $u$, if there does not exist another node with larger capacity than $u$ between $v$ and $u$ (see Figure ~\ref{example3}).
We will prove that this graph is sufficient for maintaining a heterogeneous storage network in a self-stabilizing manner and also that in this graph the degree is bounded logarithmically w.h.p..

\subsubsection{The data management layer}

We discussed above how the data is assigned to the different nodes. That is the assignment strategy we use for data in the \emph{CONE}-network.

In order to understand how the various data operations are realized in the network, we have to describe how each node maintains the knowledge about the data it
has, as well as the intervals it is responsible for.
It turns out that in order for a data item to be forwarded to the correct node, which is responsible for storing it, it suffices to contact the closest node (in terms of hash value) from the left to the data item's hash value.
That is because then, if the \emph{CONE} graph has been established, this node (for example node $u$ in  Figure ~\ref{example}) is aware of the responsible node for this data item. We call the interval between $h(u)$ and the hash value of $u$'s closest right node $I_u$. We say that $u$ is \emph{supervising} $I_u$.
We show the following theorem.


\begin{theorem}\label{theo:responsibility}
In $G^{CONE}$ a node $u$ knows all the nodes v with $R(v) \cap I_u \neq \emptyset $.
\end{theorem}

\begin{proof}
We need to show that all these nodes $R(v) \cap I_u \neq \emptyset $ are in $ S^+(u) \cup P^+(u) $ $\cup S^-(u) \cup P^-(u)$. 
W.l.o.g. let us consider only the case of $S^+(u), S^-(u)$.
Indeed, there cannot be a node at the right of $u$ ($h(u)<h(t)$) that has a responsible interval in $u$'s supervising interval and that is not in $ S^+(u)$ or $ S^-(u)$.
We will prove it by contradiction.
Let $t$ be such a node. For $t$ not to be in $ S^-(u)$ or $ S^+(u)$ there must be at least one node $v$ larger (in terms of capacity) than $t$, which is closer to $u$ than $t$ ($h(u)<h(v)<h(t)$). Then $\forall x<h(v)$ it holds that $h(v)<h(t) \implies x-h(v)<x-h(t) \implies f(x-h(v))<f(x-h(t)) $, since $f$ is increasing. Moreover, since $c(v)>c(t)$ we have $\frac{1}{c(v)}f(x-h(v))<\frac{1}{c(v)}f(x-h(t))\implies C_v(x)<C_t(x)$, so $v$ dominates $t$ for $x<h(v)$. And since $h(u)<h(v)$, it cannot be that $t$ is responsible for an interval in $I_u$, since in that region $t$ is dominated (at least) by $v$.
This contradicts the hypothesis and the proof is completed.
\end{proof}

So, the nodes store their data in the following way. If a node $u$ has a data item that falls into one of its responsible intervals, it stores in addition to this item a reference to the node $v$  that is the closest from the left to this interval. Moreover, the subinterval $u$ thinks it is responsible for (in which the data item falls) is also stored (as described in the next section, when the node's internal variables are presented).
 In case the data item is not stored at the correct node, $v$ can resolve the conflict when contacted by $u$.

Now we can discuss the functionality of the data operations.
A node has operations for inserting, deleting and searching a datum in the CONE-network.

Let us focus on $\bf{searching}$ a data item. As shown above, it suffices to search for the left closest node to the data item's hash value.
We do this by using greedy routing. Greedy routing in the \emph{CONE}-network works as follows: If a search request wants to reach some
position $pos$ in $[0,1)$, and the request is currently at node $u$, then $u$ forwards $search(pos)$ to the node $v$ in $N_u$ that is closest to $pos$, until the closest node at the left of $pos$ is reached. Then this node will forward the request to the responsible node.
A more formal definition of the greedy routing follows:

\begin{definition}
The CONE Greedy routing strategy is defined as: If operation op is to be executed at position $pos$ in $[0,1]$ and op is currently at node $u$, then $u$ forwards op to the node $v$ such that $v=\argmax \left\{h(w): w \in u.S^* \cup \left\{u\right\} \wedge h(w)< pos\right\}$ if $pos>h(u)$ or $u$ forwards op to the node $v$ such that $v=\argmin \left\{h(w): w\in u.P^* \cup \left\{u\right\} \wedge h(w)> pos\right\}$ if $pos<h(u)$. If $h(v)=h(u)$ and $pos>h(u)$, then $pos \in I_u$ and $u$ forwards $op$ to the node responsible for the subinterval containing pos. If $h(v)=h(u)$ and $pos<h(u)$ then $u$ forwards op to $u.P^*[1]$ as $pos$ is in its supervised interval.
\end{definition}

 In that way we can route to the responsible node and then get an answer whether the data item is found or not, and so the searching is realized.
Note that the $\bf{deletion}$ of a data item can be realized in the same way, only that when the item is found, it is also deleted from the responsible node.
$\bf{Inserting}$ an item follows a similar procedure, with the difference that when the responsible node is found, the data item is stored by it.

Moreover, the network handles efficiently structural operations, such as the joining and leaving of a node in the network, or the change of the capacity of
a node. Since this handling falls into the analysis of the self-stabilization algorithm, we will discuss the network operations in Section ~\ref{alg}, where we also formally analyze the algorithm.

It turns out that a single data or network operation (i.e greedy routing) can be realized in a logarithmic number of hops in the \emph{CONE}-network, and this happens due  to
the structural properties of the network, which we discuss in the next section, where we also show that the degree of the \emph{CONE}-network is logarithmic.

\subsection{Structural Properties of a Cone Network}

In this section we show that the degree of a node in a stable CONE-network is bounded by $\mathcal O(\log n)$ w.h.p, and hence the information stored by each node (i.e the number of nodes which it maintains contact to, $|E_e(u)|$) is bounded by $\mathcal O(\log n + |\text{amount of data stored in a node}|)$ w.h.p..

First we show following lemma:

\begin{lemma}\label{log}
In a stable CONE network for each $u\in V$, $| S^+(u)|$ and $| P^+(u)|$ in $\mathcal O(\log n)$ w.h.p.
\end{lemma}

\begin{proof}
For an arbitrary $u\in V$ let $W=\left\{w_1,w_2\cdots w_k\right\}\cup \left\{w_0 =u\right\}= S^+(u)\cup \left\{u\right\}$ and let $W$ be sorted by ids in ascending order, such that $h(w_i)<h(w_{i+1})$ for all $1\leq i<k$. Furthermore, let $\hat{W}(w_i)=\left\{w\in V:h(w)>h(w_i) \wedge c(w)>c(w_i)\right\}$ be the set of all nodes with larger ids and larger capacities than $w_i$.
So, the determination of $W$ is done by continuously choosing the correct $w_i$ out of $\hat{W}(w_{i-1})$, when $w_1,w_2,...w_{i-1}$ are already chosen.
In this process, each time a $w_i$ is determined, the number of nodes from which $w_{i+1}$ can be chosen is getting smaller, since the nodes at the left of $w_i$ as well as the nodes with smaller capacity than $w_i$ can be excluded.
We call the choice of $w_i=\hat{w_j}$ \emph{good}, if $|\hat{W}(w_{i-1})|>2|\hat{W}(w_i)|$, i.e. the number of remaining nodes in $\hat{W}(w_i)$ is (at least) halved. Let $|\hat{W}(w_{0}|=m=\mathcal O(\log n)$. Since the id/position for each node is assigned uniformly at random, we can easily see that Pr[$w_i$ is a good choice]$= \frac{1}{2}, \forall i\geq 1$.
Then after a sequence of $i$ choices that contains $\log m$ good choices the remaining set $\hat{W}(w_i)$ is the empty set. Thus there can not be more than $\log m$ good choices in any sequence of choices.
So, what we have now is a random experiment, that is described by the random variable $k$, that is equal to the number choices we must make, until we managed to have made $\log m$ good ones. Then the random variable $k$ follows the negative binomial distribution.
In order to bound the value of $k$ from above we apply the following tail bound for negative binomially distributed random variables shown in \cite{NegBi}, derived by using a Chernoff bound:

\begin{claim}
Let $Y$ have the negative binomial distribution with parameters $s$ and $p$, i.e. with probability $p$ there is a success and $Y$ is equal to the number of trials needed for $s$ successes. Pick $\delta \in [0,1]$ and set $l=\frac{s}{(1-\delta)p}$. Then $Pr[Y>l]\leq exp(\frac{-\delta^{2}s}{3(1-\delta)})$
\end{claim}

We apply this claim with $p=\frac{1}{2}$ and $s=\log m$ and we pick $\delta=\frac{7}{8}$. Then $Pr[Y>16\log m]\leq \exp(\frac{-\delta^{2}2\log m}{3(1-\delta)})<m^{- 2}$. Thus with probability at least $1-m^{- 2}$, $k=\mathcal O(\log m)$ as $m=\mathcal O(n)$ also $k=\mathcal O(\log n)$.
\end{proof}

\begin{lemma}
In a stable CONE network for each $u\in V$, $E[| S^-(u)|]$ and $E[ |P^-(u)|]$ are $\mathcal O(1)$ and $| S^-(u)|$ and $| P^-(u)|$ are $\mathcal O(\log n)$ w.h.p..
\end{lemma}

\begin{proof}
W.l.o.g. we consider only $E[| S^-(u)|]$ and $|S^-(u)|$ in the proof.
For each node $x$ being in an interval $S^-(u)$ it holds $u=pred^+_1(x)$. But since each node has (at most) one $pred^+_1$, the sum over all $S^-(v)$, $\forall v \in V$ must be (at most) $n$. So we have $\sum_{\forall v \in V} S^-(v)=n$, so  $E[\sum_{\forall v \in V} S^-(v)]=n$ $\Rightarrow$ $\sum_{\forall v \in V}E[ S^-(v)]=n $.
That means for a node $u$, $E[ S^-(u)]=1$.

Now we consider the second part of the statement.
Let $w$ be the direct right neighbor of $u$, i.e. the first (from the left) node in $S^-(u)$ ($S^-(u)[1]$). Then we can observe that every node in $S^-(u)$ (expect $w$) must be in $S^+(w)$. Let us assume a node $x$ is in $S^+(w)$  but not in $S^-(u)$, then there must be another node $y: h(u)<h(y)<h(x)$ and $c(y)>c(x)$, such that $y \in S^-(u)$. But then $y$ would be also in  $S^+(w)$ instead of $x$.
So, we contradicted this scenario. As a consequence  $S^-(u)/\{w\} \subset S^+(w)$, but we already shown that $|S^+(w)|<\log n$ w.h.p., from which follows that $|S^-(u)|<\log n$ w.h.p..
\end{proof}

Combining the two lemmas we get the following theorem.

\begin{theorem}\label{theo:degree}
The degree of each node in a stable CONE network is $\mathcal O(\log n)$ w.h.p.
\end{theorem}

Additionally to the nodes in $ S^+(u)$, $ S^-(u)$, $ P^+(u)$ and $ P^-(u)$ that lead to the degree of $\mathcal O(\log n)$ w.h.p. a node $u$ only stores references about the closest nodes left to the intervals it is responsible for, where it actually stores data. A node $u$ stores at most one reference and one interval for each data item. Thus the storage only has a logarithmic overhead for the topology information and the following theorem follows immediately.

\begin{theorem}\label{theo:storage}
In a stable CONE network each node stores at most $\mathcal O(\log n + |\text{amount of data stored in a node}|)$ information w.h.p.
\end{theorem}

Once the CONE network $G^{CONE}$ is set up, it can be used as an heterogeneous storage system supporting inserting, deleting and searching for data.
The CONE Greedy routing implies the following bound on the diameter:

\begin{lemma}\label{lem:routing}
CONE Greedy routing takes on a stable CONE network w.h.p. no more than a logarithmic number of steps, i.e. the diameter of a CONE network is $\mathcal O(\log n)$ w.h.p..
\end{lemma}

\begin{proof}
This follows directly from Lemma \ref{log}, where we showed that each node $u$ has w.h.p. a logarithmic number of nodes in $ P^+(u)/ S^+(u)$, which means it has a logarithmic distance to the node with the greatest capacity, and vice versa, which means that the node with the greatest capacity has logarithmic distance to every node in the network. The proof for the CONE Greedy routing follows from a generalization of this observation. If an operation $op$ with position $pos \in I_u$ is currently at node $v$ w.l.o.g. we assume $h(u)>h(v)$, then $op$ is forwarded at most $\mathcal O(\log n)$ times w.h.p. (along nodes in $ S^+(v)$) to a node $w$ such that $w\in  S^+(v)$  and further $\mathcal O(\log n)$ times w.h.p. (along nodes in $ S^-(u)$) from $w$ to $u$.
\end{proof}

\begin{figure}[htb]
\caption{In this example, the case of the linear capacity functions is presented.
Concerning the intervals lying between $u$ and $v$,
we can see by the coloring which interval is assigned to which node (the one having the lowest capacity function value at that interval). According to the \emph{CONE}-graph, $u$
must be aware of all these nodes ($v,w$ and $x$).}

\includegraphics[scale=0.4]{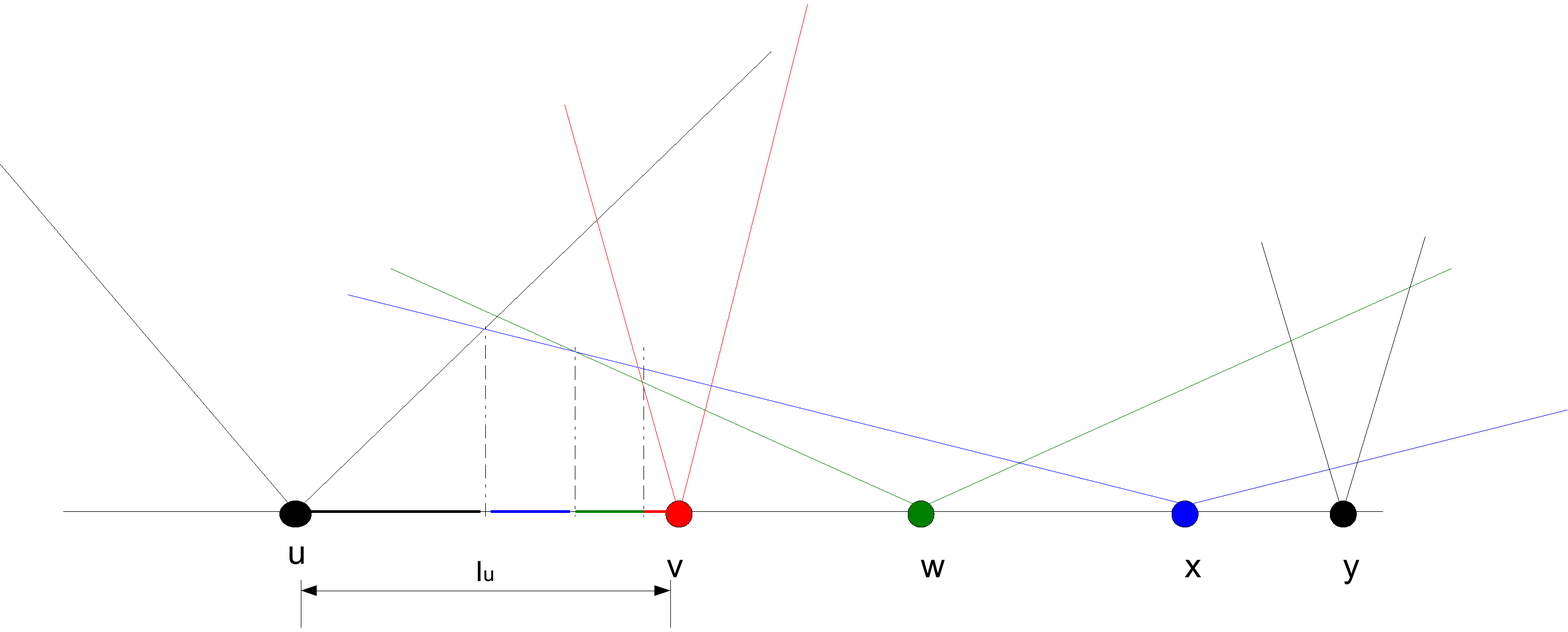}\label{example}
\end{figure}

\begin{figure}[htb]

\caption{In this example, the size of the capacity of a node is symbolized by the height of its green column (i.e. the larger the capacity the higher the column). So, for example in this case $u$ is aware of $v,w$ and $x$. In fact, $S^+(u)=\{w,x\},S^-(u)=\emptyset,P^+(u)=\emptyset, P^-(u)=\{v \}$   }

\includegraphics[scale=0.4]{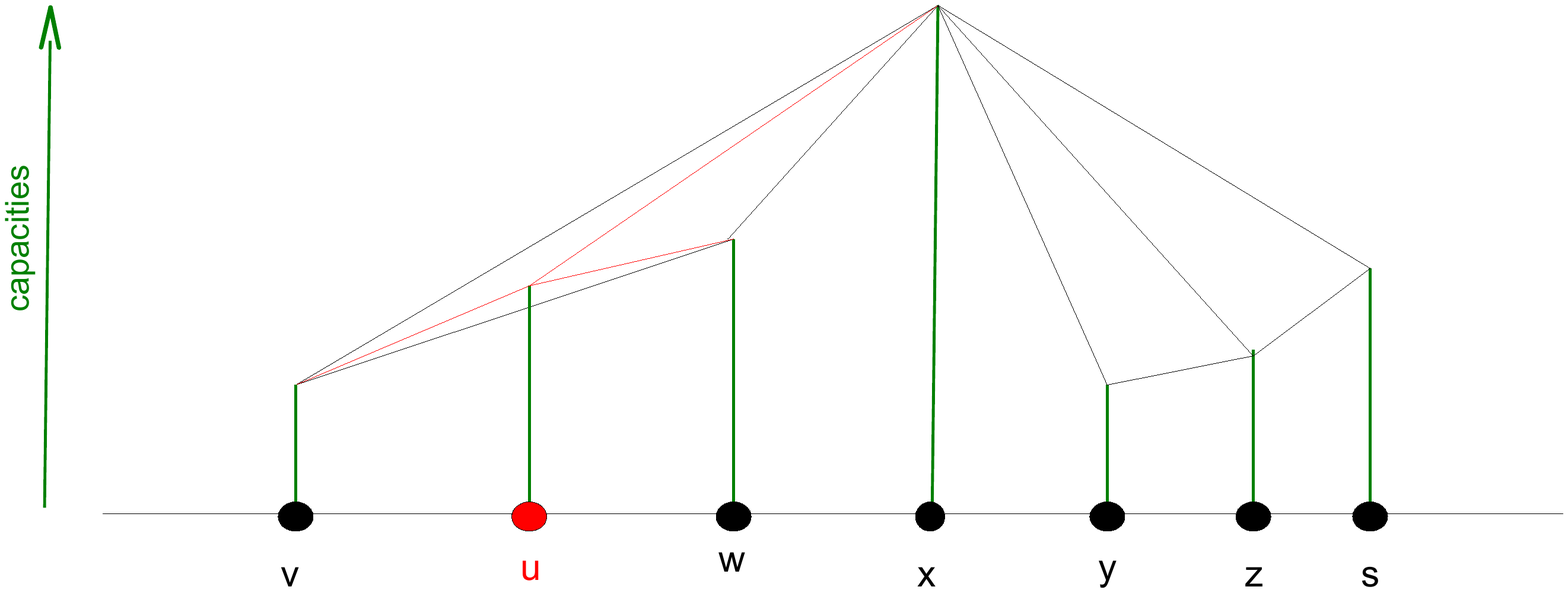}\label{example3}
\end{figure}
\section{Self-Stabilization Process}\label{alg}

\subsection{Topological Self-stabilization}
We now formally describe the problem of topological self-stabilization. In topological self-stabilization the goal is to state a protocol \emph{P} that \emph{solves} an overlay problem \emph{OP} starting from an initial topology of the set \emph{IT}. A protocol is \emph{unconditionally} self-stabilizing if \emph{IT} contains every possible state. Analogously a protocol is \emph{conditionally} self-stabilizing if \emph{IT} contains only states that fulfill some conditions. For topological self-stabilization we assume that \emph{IT} contains any state as long as $G^{IT}=\left(V,E^{IT}\right)$ is weakly connected, i.e. the combined knowledge of all nodes in this state covers the whole network, and there are no identifiers that don't belong to existing nodes in the network.
The set of target topologies defined in \emph{OP} is given by $OP=\left\{G^{OP}_e=\left(V, E^{OP}_e\right)\right\}$, i.e. the goal topologies of the overlay problem are only defined on explicit edges and $E^{OP}_i$ can be an arbitrary (even empty) set of edges. We also call the program states in \emph{OP} \emph{legal states}. We say a protocol \emph{P} that solves a problem \emph{OP} is topologically self-stabilizing if for \emph{P} \emph{convergence} and \emph{closure} can be shown. \emph{Convergence} means that \emph{P} started with any state in \emph{IT} reaches a legal state in \emph{OP}. \emph{Closure} means that \emph{P} started in a legal state in \emph{OP} maintains a legal state.
For a protocol \emph{P} we assume that there are no oracles available for the computation. In particular we assume there to be no \emph{connection oracle}, that can connect disconnected parts of the network, no \emph{identifier detector}, that can decide whether an identifier belongs to an existing node or not, and no \emph{legal state detector}, that can decide based on global knowledge whether the system is in a legal state or not.
With these assumptions our model complies with the \emph{compare-store-send} program model in \cite{NNS11:corona} in which protocols do not manipulate the internals of the nodes' identifiers. For our modified model the impossibility results of \cite{NNS11:corona} still hold such as Lemma 1 that states if the graph $G=(V,E=E_i \cup E_e)$ is initially disconnected, then the graph is disconnected in every state of the computation. Furthermore Theorem 1 states that if the goal topology is a single component a program only solves the problem if the initial graph is weakly connected.


\subsection{Formal problem definition and notation}\label{problem}


Now we define the problem we solve in this paper in the previously introduced notation. We provide a protocol \emph{P} that solves the overlay problem \emph{CONE} and is topologically self-stabilizing.


In order to give a formal definition of the edges in $E_e$ and in $E_i$ we firstly describe which internal variables are stored in each node $u$, i.e. which edges exist in $E_e$:

\begin{itemize}\itemsep0.1pt

\item $u.S^+=\left\{v\in N_u: h(v)>h(u) \wedge c(v)>c(u) \wedge \forall w\in N_u: h(v)>h(w)>h(u) \implies c(v)>c(w)\right\}$

\item $u.succ^+_1=\argmin \left\{h(v): v \in u.S^+ \right\}$: The first node to the right with a larger capacity than $u$

\item $u.P^+=\left\{v\in N_u: h(v)<h(u) \wedge c(v)>c(u) \wedge \forall w\in N_u: h(v)<h(w)<h(u) \implies c(v)>c(w)\right\}$

\item $u.pred^+_1=\argmax \left\{h(v): v \in u.P^+\right\}$: The first node to the left with a larger capacity than $u$

\item $u.S^-=\left\{v\in N_u: h(v)>h(u) \wedge c(v)<c(u) \wedge \forall w\in N_u: h(v)>h(w)>h(u) \implies c(v)>c(w)\right\}$

\item $u.P^-=\left\{v\in N_u: h(v)<h(u) \wedge c(v)<c(u) \wedge \forall w\in N_u: h(v)<h(w)<h(u) \implies c(v)>c(w)\right\}$

\item $u.S^*=\left\{u.S^- \cup \left\{u.succ^+_1\right\}\right\}$: the set of right neighbors that $u$ communicates with. We assume that the nodes are stored in ascending order so that $h(u.S^*[i])<h(u.S^*[i+1])$. If there are $k$ nodes in $u.S^*$ then $u.S^*[k]=u.succ^+_1$. 

\item $u.P^*=\left\{u.P^- \cup \left\{u.pred^+_1\right\}\right\}$: the set of left neighbors that $u$ communicates with. We assume that the nodes are stored in descending order so that $h(u.P^*[i])>h(u.P^*[i+1])$ If there are $k$ nodes in $u.P^*$ then $u.P^*[k]=u.pred^+_1$.

\item $u.DS$ the data set, containing all intervals $u.DS[i]=[a,b]$, for which $u$ is responsible and stores actual data $u.DS[i].data$.  
Additionally for each interval a reference $u.DS[i].ref$ to the supervising node is stored
\end{itemize}

Additionally each node stores the following variables :
\begin{itemize}\itemsep0.1pt

\item $\tau$: the timer predicate that is periodically true

\item $u.I_u$: the interval between $u$ and the successor of $u$. $u$ is supervising $u.I_u$.

\item $m$: the  message in $Ch_u$ that now received by the node.

\end{itemize}

\begin{definition}\label{def:valid}
We define a valid state as an assignment of values to the internal variables of all nodes so that the definition of the variables is not violated, e.g. $u.S^+$ contains no nodes $w$ with $h(w)<h(u)$ or $c(w)<c(u)$ or $h(u)<h(v)<h(w)$ and $c(v)>c(w)$ for any $v\in N_u$.
\end{definition}

Now we can describe the topologies in the initial states and in the legal stable state. Let $IT=\left\{G^{IT}=(V,E_{IT}=E^{IT}_e \cup E^{IT}_i):G^{IT}\text{ is weakly connected}\right\}$ and let $CONE=\left\{G^{C}=\left(V, E^{C}\right)\right\}$, such that for $E^{C}$ the following conditions hold:

\begin{itemize}
  \item $E^{C}=E_e-\left\{(u,v): v\in u.DS\right\}$
  \item $E^{C}$ is in a valid state
  \item $E^{C}=E^{CONE}$
\end{itemize}

Note that we assume $E_e$ to be a multiset, i.e in $E^{C}$ an edge $(u,v)$ might still exists, although $v \in u.DS$ if e.g. $v \in u.S^+$.
Further note that, in case the network has stabilized to a \emph{CONE}-network, it holds for every node that $u.S^+=S^+(u),u.P^+=P^+(u),u.S^-=S^-(u)$ and $u.P^-=P^-(u)$.

\subsection{Algorithm}\label{app:alg}
In this section we give a description of the the distributed algorithm. The
algorithm is a protocol that each node executes based on its own node and
channel state. The protocol contains periodic actions that are executed if
the timer predicate $\tau$ is true and actions that are executed if the node
receives a message $m$.
In the periodic actions each node performs a consistency check of its
internal variables, i.e. are all variables valid according to Definition
~\ref{def:valid}. If some variables are invalid, the nodes causing this
invalidity are delegated. By \emph{delegation} we mean that node $u$
delegates a node $v: h(v)>h(u)$ (resp. $h(v)<h(u)$) to the node $w'=\argmax \left\{h(w) : w \in
u.S^* \wedge h(w)<h(v)\right\}$ (resp. $w'=\argmin \left\{h(w): w \in u.P^* \wedge h(w)>h(v)\right\}$) by a message
$m=(build-triangle,v)$ to $w'$. The idea behind the delegation is to forward nodes closer to their correct position, so that the sorted list (and the \emph{CONE}-network) is formed.
In order for a node $u$ to maintain valid lists ($u.S^+,u.P^+$), it makes a periodic check of its lists with its neighbors in $u.S^-,u.P^-$, where the lists are compared, so that inconsistencies
are repaired. Moreover a node checks whether $u.S^-/u.P^-$ are valid and introduces to them their closest larger right/left neighbors (from $u$'s perspective).
Unnecessary (for these lists) nodes are delegated.
We show later in the analysis section that this process leads to the construction of the correct lists by each node and thus to the \emph{CONE}-network. 
Furthermore in the periodic actions each node introduces itself to its successor and predecessor $u.S^*[1]$ and $u.P^*[1]$ by a message $m=(build-triangle,u)$. Also each pair of nodes in $u.P^*$ and $u.S^*$ with consecutive ids is introduced to each other. $u$ also introduces the nodes $u.succ^+_1$ and $u.pred^+_1$ to each other by messages of type $build-triangle$. By this a \emph{triangulation} is formed by edges
$(u,u.pred^+_1), (u,u.succ^+_1), (u.succ^+_1,u.pred^+_1)$ (see Figure ~\ref{example3}).
To establish correct $P^+$ and $S^+$ lists in each node, a node $u$ sends its
$u.P^+$ (resp. $u.S^+$) list periodically to all nodes $v$ in $u.S^-$ (resp.
$u.P^-$) by a message $m=(list-update,u.S^+ \cup \left\{u\right\})$ (resp.
$m=(list-update,u.P^+\cup \left\{u\right\})$) to $v$.
The last action a node periodically executes is to send a message to each
reference in $u.DS$ to check whether $u$ is responsible for the data in the
corresponding interval $[a,b]$ by sending a message
$m=(check-interval,[a,b],u)$.

If the message predicate is true and $u$ receives a message $m$, the
action $u$ performs depends on the type of the message. If $u$ receives a
message $m=(build-triangle,v)$ $u$ checks whether $v$ has to be included in
it's internal variables $u.P^+$, $u.S^+$, $u.P^-$ or $u.S^-$. If $u$ doesn't
store $v$, $v$ is delegated. If $u$ receives a message
$m=(list-update,list)$, $u$ checks whether the ids in $list$ has to be
included in it's internal variables $u.P^+$, $u.S^+$, $u.P^-$ or $u.S^-$. If $u$ doesn't store a node $v$ in $list$, $v$ is delegated. If $u$ stores a node
$v$ in $u.S^+$ (resp. $u.P^+$) that is not in $list$, $v$ is also delegated as
it also has to be in the list of $u.pred^+_1$ (resp. $u.succ^+_1$). The remaining
messages are necessary for the data management.

If $u$ receives a message $m=(check-interval,[a,b],v)$ it checks whether $v$
is in $u.S^+$ or $u.P^+$ or has to be included, or delegates $v$. Then $u$
checks whether $[a,b]$ is in $u.I_u$ and if $v$ is responsible for $[a,b]$. If
not, $u$ sends a message $m=(update-interval,IntervalSet)$ to $v$ containing
a set of intervals in $[a,b]$ that $v$ is not responsible for and references
of the supervising nodes.
If $u$ receives a message
$m=(update-interval,IntervalSet)$ it forwards all data in intervals in
$IntervalSet$ to the corresponding references by a message
$m=(forward-data,data)$. If $u$ receives such a message it checks whether
the data is in its supervised interval $u.I_u$. If not $u$ forwards the data
according to a greedy routing strategy, if $u$ supervises the data it sends
a message $m=(store-data,data,u)$ to the responsible node. If $u$ receives
such a message it inserts the data, the interval and the corresponding
reference in $u.DS$. Note that no identifiers are ever deleted, but always
stored or delegated. This ensures the connectivity of the network.

In the following we give a description of the protocol executed by each node
in pseudo code.

\subsection{Pseudo code}


\begin{algorithm}[H]
Periodic actions including a consistency check, where all list $u.P^+,u.S^+,u.S^-,u.P^-$ are checked and invalid nodes are delegated like in the ListUpdate/BuildTriangle operation. Furthermore each node sends its lists $u.P^+,u.S^+$ to the  next smaller nodes and $u.S^+$ to its direct left neighbor. Finally in the BuildTriangle() $u$ introduces itself to its neighbors and neighbored nodes in $u.P^*$ and $u.S^*$ and $u.pred^+_1$ and $u.succ^+_1$ to each other and checks whether the information in $u.DS$ is still up to date.

\begin{algorithmic}
\State $\tau \rightarrow$ \Comment{periodic actions}
\State Consistency check for $u.P^+$, $u.S^+$, $u.S^-$, $u.P^-$
\ForAll{$w \in u.S^-$}
\State send m=(list-update,$u.P^+\cup\left\{u\right\}$) to $w$
\EndFor
\ForAll{$w \in u.P^-$}
\State send m=(list-update,$u.S^+\cup\left\{u\right\}$) to $w$
\EndFor
\State BuildTriangle()
\State CheckDataIntervals()
\end{algorithmic}
\end{algorithm}

\begin{algorithm}
\begin{algorithmic}

\Function{BuildTriangle}{node x}
\If{$x = \emptyset$} \Comment{periodic introduction of nodes and $u$ itself}
\ForAll{$w \in u.S^-$}
\State $v^{-}(w)=\argmax \left\{h(v): v\in u.S^* \wedge h(v)< h(w) \right\}$
\State send m=(build-triangle,w) to $v^{-}(w)$ and  m'=(build-triangle,$v^{-}(w)$) to $w$
\EndFor
\ForAll{$w \in u.P^-$}
\State $v^{+}(w)=\argmin \left\{h(v):v \in u.P^*\wedge h(v)>h(w)\right\}$
\State send m=(build-triangle,$w$) to $v^{+}(w)$ and  m'=(build-triangle,$v^{+}(w)$) to $w$
\EndFor
\ForAll{$w \in u.S^- \cup u.P^- \cup \left\{u.succ^+_1,u.pred^+_1\right\}$}
\State send m=(build-triangle,$u$) to $w$
\EndFor
\State send m=(build-triangle,$u.pred^+_1$) to $u.succ^+_1$ and  m'=(build-triangle,$u.succ^+_1$) to $u.pred^+_1$
\Else \Comment{demand action by a received node id}
\If{$c(x)>c(u) \wedge h(u)<h(x)<h(u.succ^+_1)$}
\State send m=(buildtriangle,$u.succ^+_1$) to $x$
\State $u.succ^+_1=x$
\ElsIf{$c(x)>c(u) \wedge h(u)>h(x)>h(u.pred^+_1)$}
\State send m=(buildtriangle,$u.pred^+_1$) to $x$
\State $u.pred^+_1=x$
\ElsIf {$h(x)>h(u)$}
\State calculate $S^-_{tmp}$ out of $u.S^-$ and $x$
\ForAll{$w\in (u.S^- \cup \left\{x\right\})-S^-_{tmp}$}
\State $v^{-}(w)=\argmax \left\{h(v):v\in S^-_{tmp} \cup \left\{succ^+_1\right\}\wedge h(v)<h(w)\right\}$
\State send m=(build-triangle,$w$) to $v^{-}(w)$
\EndFor
\State $u.S^-=S^-_{tmp}$
\ElsIf {$h(x)<h(u)$}
\State calculate $P^-_{tmp}$ out of $u.P^-$ and $x$
\ForAll{$w\in (u.P^- \cup \left\{x\right\})-P^-_{tmp}$}
\State $v^{+}(w)=\argmin \left\{h(v):v\in P^-_{tmp} \cup \left\{pred^+_1\right\}\wedge h(v)>h(w)\right\}$
\State send m=(build-triangle,$w$) to $v^{+}(w)$
\EndFor
\State $u.P^-=P^-_{tmp}$
\EndIf
\EndIf
\EndFunction
\end{algorithmic}
\end{algorithm}

\begin{algorithm}

\begin{algorithmic}
\Function{ListUpdate}{List}
\State $LList^+={z \in List: h(z)<h(u) \wedge c(z)>c(u)}$ \Comment{candidates for $u.P^+$}
\State $LList^-={z \in List: h(z)<h(u) \wedge c(z)<c(u)}$ \Comment{candidates for $u.P^-$}
\State $RList^+={z \in List: h(z)>h(u) \wedge c(z)>c(u)}$ \Comment{candidates for $u.S^+$}
\State $RList^-={z \in List: h(z)>h(u) \wedge c(z)<c(u)}$ \Comment{candidates for $u.S^-$}
\State calculate $P^+_{tmp}$ out of $u.P^+$ and $LList^+$ \Comment{calculate new lists and delegate all nodes not stored in the new lists}
\State $Z=(u.P^+-LList^+) \cup ((u.P^+ \cup LList^+)-P^+_{tmp})$
\If{$u.P^+\neq P^+_{tmp}$}
\ForAll{$z \in Z$}
\State send m=(build-triangle,z) to $P^+_{tmp}[1]$
\EndFor
\State $u.P^+=P^+_{tmp}$
\EndIf
\State calculate $S^+_{tmp}$ out of $u.S^+$ and $RList^+$
\State $Z=(u.S^+-RList^+) \cup ((u.S^+ \cup RList^+ )- S^+_{tmp})$
\If {$u.S^+\neq S^+_{tmp}$}
\ForAll{$z\in Z$}
\State send m=(build-triangle,$z$) to $S^+_{tmp}[1]$
\EndFor
\State $u.S^+=S^+_{tmp}$
\EndIf
\State calculate $P^-_{tmp}$ out of $u.P^-$ and $LList^-$
\ForAll{$w\in (u.P^- \cup LList^-)-P^-_{tmp}$}
\State $v^{+}(w)=\argmin \left\{h(v):v\in P^-_{tmp} \cup \left\{pred^+_1\right\}\wedge h(v)>h(w)\right\}$
\State send m=(build-triangle,$w$) to $v^{+}(w)$
\EndFor
\State $u.P^-=P^-_{tmp}$
\State calculate $S^-_{tmp}$ out of $u.S^-$ and $RList^-$
\ForAll{$w\in (u.S^- \cup RList^-)-S^-_{tmp}$}
\State $v^{-}(w)=\argmax \left\{h(v):v\in S^-_{tmp} \cup \left\{succ^+_1\right\}\wedge h(v)< h(w)\right\}$
\State send m=(build-triangle,$w$) to $v^{-}(w)$
\EndFor
\State $u.S^-=S^-_{tmp}$
\EndFunction
\end{algorithmic}
\end{algorithm}


\begin{algorithm}
A node checks for each interval it is responsible for, if this is really the case.

\begin{algorithmic}
\Function{CheckDataIntervals}{}
\ForAll{u.DS[i]}
\State send m=(check-interval,$[a,b]=u.DS[i]$,$u$) to $u.DS[i].ref$
\EndFor
\EndFunction
\end{algorithmic}
\end{algorithm}

\begin{algorithm}
A node receiving a check-interval message, checks if the node which sent it is really responsible for the interval [a,b].

\begin{algorithmic}
\Function{CheckInterval}{[a,b],x}
\If{$x \not\in u.P^+ \cup u.S^+ \cup u.S^-$}
\State BuildTriangle(x)
\EndIf
\State IntervalSet := $\emptyset$
\State i:=1
\If{$a <h(u)$}
\State IntervalSet[i]=$[a,h(u)] \cap [a,b]$ \Comment{The interval begins left of $u$, so $u$ can't be the supervising node for the whole interval}
\State IntervalSet[i].ref=$u.P^*[1]$
\State i:=i+1
\EndIf
\If{$b>u.S^*[1]$}
\State IntervalSet[i]=$[u.S^*[1],b]\cap [a,b]$ \Comment{The interval ends right of $u.S^*[1]$, so $u$ can't be the supervising node for the whole interval}
\State IntervalSet[i].ref=$u.S^*[1]$
\State i:=i+1
\EndIf
\State $[c,d]:=I_u(x)$          \Comment{$I_u(x)$ is  the subinterval of $u.I_u$ for which $x$ is responsible for}
\State [e,f]:=$([a,b]\cap u.I_u)/I_u(x)$
\If{$e<c$}
\State IntervalSet[i]=$[e,c]\cap [a,b]$ \Comment{$u$ as the supervising node, knows other nodes responsible for parts of the interval}
\State IntervalSet[i].ref=u
\State i:=i+1
\EndIf
\If{$f>d$}
\State IntervalSet[i]=$[d,f]\cap [a,b]$ \Comment{$u$ as the supervising node, knows other nodes responsible for parts of the interval}
\State IntervalSet[i].ref=u
\EndIf
\State send m=(update-interval,IntervalSet) to $x$
\EndFunction
\end{algorithmic}
\end{algorithm}

\begin{algorithm}
By receiving an update-interval message, a node updates the lists of intervals which it is responsible for, and forwards the data in the
deleted intervals to another node, who is possibly responsible.

\begin{algorithmic}
\Function{UpdateInterval}{IntervalSet}
\ForAll {$[a,b]\in IntervalSet$}
\ForAll{$[c,d]\in u.DS$}
\ForAll {$[e,f]\in\left\{[c,d]-[a,b]\right\}$}
\State l:=$|u.DS|$
\State u.DS[l+1]=[e,f]
\State u.DS[l+1].ref=[c,d].ref \Comment{references are set to the new supervising node}
\State u.DS:=$u.DS-\left\{[c,d]\right\}$
\EndFor
\ForAll {$data \in [c,d]\cap [a,b]$}
\State send m=forward-data(data) to  $[a,b].ref$ \Comment{data $u$ seems not to be responsible for or for that the reference changed is deleted}
\State delete(data)
\State BuildTriangle([a,b].ref) \Comment{references supervising nodes are forwarded to maintain connectivity}
\EndFor
\EndFor
\EndFor
\State UpdateDS() \Comment{Delete all intervals without data, forward the references of the deleted intervals, unite all consecutive intervals with the same reference}
\EndFunction
\end{algorithmic}
\end{algorithm}

\begin{algorithm}
By receiving a forward-data message, a node checks if it knows which node is responsible for the data it received, and sends a store-data message to it, in the other case it also forwards the data.

\begin{algorithmic}
\Function{ForwardData}{data}
\If {$data.id \not\in u.I_u$}
\If {$data.id \in [u.P^*[1]],u]$}
\State send m=(forward-data(data)) to $u.P^*[1]$
\Else
\State send m=(forward-data(data) to
\State $w: (h(u)<h(w)<data.id \vee h(u)>h(w)>data.id)\wedge  |data.id-h(w)| = \min_{y \in u.P^*\cup u.S^*}\{|data.id-h(y)| \},  $
\EndIf
\Else
\State send m=(store-data,data,$I_u(v)$,u) to $v: data.id\in I_u(v)$
\EndIf
\EndFunction
\end{algorithmic}
\end{algorithm}
\begin{algorithm}
Storing the data received from the node supervising the corresponding interval.

\begin{algorithmic}
\Function{StoreData}{data,interval,x}
\If{$\exists i: interval =u.DS.i$}
\State u.DS[i].data := $u.DS[i].data \cup data.id \in u.DS[i]$
\State BuildTriangle(u.DS[i].ref)
\State u.DS[i].ref=x
\Else
\State l:=$|u.DS|$
\State u.DS[l+1]=interval
\State u.DS[l+1].ref=x
\State u.DS[l+1].data=data
\EndIf
\EndFunction
\end{algorithmic}
\end{algorithm}

%
\begin{algorithm}
\begin{algorithmic}
\State $m \in Ch_u \text{ received by u} \rightarrow$ \Comment{demand actions depending on the received message}
\If{m=(list-update,List)}
\State ListUpdate(List)
\ElsIf{m=(build-triangle,x)}
\State BuildTriangle(x)
\ElsIf{m=(check-interval,[a,b],x)}
\State CheckInterval([a,b],x)
\ElsIf{m=(update-interval,IntervalSet)}
\State UpdateInterval(IntervalSet)
\ElsIf{m=(forward-data,data,boolean)}
\State ForwardData(data,boolean)
\ElsIf{m=(store-data,data,interval,x)}
\State StoreData(data,interval,x)
\EndIf

\end{algorithmic}
\end{algorithm}
\section{Correctness}
In this section we show the correctness of the presented algorithm. We do this by showing that by executing our algorithm any weakly connected network eventually converges to a CONE network and once a CONE network is formed it is maintained in every later state. We further show that in a CONE network the data is stored correctly.

\subsection{Convergence}
To show convergence we will divide the process of convergence into several phases, such that once one phase is completed its conditions will hold in every later program state. For our analysis we additionally define $E(t)$ as the set of edges at time $t$. Analogous $E_e(t)$ and $E_i(t)$ are defined.
We show the following theorem. 

\begin{theorem}\label{theo:convergence}
If $G=(V,E) \in IT$ at time $t$ then eventually at a time $t'>t$ $G^{CONE}\subset G_e(t')$.
\end{theorem}

We divide the proof into 3 phases. First we show the preservation of the connectivity of the graph, then we show the convergence to the sorted list and eventually the convergence to the \emph{CONE}-network.

\subsubsection{Phase 1: Connectivity}
In the first phase we will show that the protocol keeps the network weakly connected and eventually forms a network that is connected by edges $(x,y)\in E_e$ such that $y\in x.P^+ \cup x.S^+ \cup x.S^- \cup x.P^-$ and edges $(x,y)\in E_i$ such that $m=(build-triangle,y) \in Ch_x$.

\begin{lemma}\label{lem:GConnectivity}
Any graph which is weakly connected due to edges in  $E_e \cup E_i$ stays weakly connected according to the given protocol, i.e if $E(t)$ is weakly connected then $\forall t'>t$, $E(t')$ is also weakly connected.
\end{lemma}

\begin{proof}
We show that for each existing edge $(x,y)\in E(t)$ either the edge remains and $(x,y)\in E(t+1)$ or a path connecting $x,y$ exists. Obviously $x$ and $y$ stay weakly connected as long as an edge $(x,y)$ exists. We therefore assume $(x,y) \in E(t)$ and $(x,y)\notin E(t+1)$. If $y$ is stored in an internal variable of $x$ then there can be the following cases:

\begin{itemize}
  \item $y \in x.P^+ \cup x.S^+$ at time $t$, then $y$ is delegated to $x.pred^+_1$ (resp. $x.succ^+_1$) and $x$ and $y$ stay connected over the edges $(x,x.pred^+_1) \in E_e(t+1)$ and $(x.pred^+_1,y)\in E_i(t+1)$.
  \item $y \in x.S^- \cup x.P^-$ at time $t$, then $y$ is delegated to $v^{-}(y)=\argmax \left\{h(v) : v \in x.S^* \wedge h(v)< h(y) \right\}$ (resp. $v^{+}(y)$) and $x$ and $y$ stay connected.
  \item $y \in x.DS[i].ref$ at time $t$ then x has received an interval-update message with a new reference for the data in $x.DS[i]$ or the data is deleted. Then in both cases $x$ delegates $y$ to $v^{-}(y)$ (resp. $v^{+}(y)$) and $x$ and $y$ stay connected.
\end{itemize}

If $y$ is stored in an incoming message $m$ in $Ch_x$. When $m$ is received then there can be the following cases:

\begin{itemize}
  \item $y$ is in a list in a list-update message. Then either $y$ is stored in a new list $x.P^+,x.S^+,x.S^-,x.P^-$ or $y$ is delegated to $v^{-}(y)$ (resp. $v^{+}(y)$) and $x$ and $y$ stay connected.
  \item $y$ is the node sending a check-interval message. Then either $y$ is stored in $x.P^+, x.S^+$ or delegated to $v^{-}(y)$ (resp. $v^{+}(y)$) and $x$ and $y$ stay connected.
  \item $y$ is a reference in an interval-update message, then either $y$ is stored as a new reference for some data or if there is no corresponding data $y$ is delegated to $v^{-}(y)$ (resp. $v^{+}(y)$) and $x$ and $y$ stay connected.
  \item $y$ is a reference in an store-data message, then $y$ is stored as a new reference in $x.DS$.
  \item $y$ is the id in a build-triangle message, then $y$ is either stored in one of the lists $x.P^+,x.S^+,x.S^-,x.P^-$ or $y$ is delegated to $v^{-}(y)$ (resp. $v^{+}(y)$) and $x$ and $y$ stay connected.
\end{itemize}
\end{proof}

\begin{definition}\label{def:TriangulationGraph}
Let $E^{Tri}_e=\left\{(x,y):y\in x.P^+ \cup x.S^+ \cup x.P^- \cup x.S^-\right\}$ and  $E^{Tri}_i=$ \newline $\left\{(x,y):  m=(build-triangle,y)\in Ch_x\right\} $ we then define the graph $G^{Tri}=\left(V,E^{Tri}=E^{Tri}_e \cup E^{Tri}_i\right)$ as the \emph{triangulation graph}.
\end{definition}

\begin{lemma} \label{lem:TriConnectivity}
If $x$ and $y$ are connected in $G^{Tri}$ at time $t$ then they will be weakly connected at every time $t'>t$.
\end{lemma}

\begin{proof}
Again we will consider every possible edge $(x,y)$ in $G^{Tri}$ and show that $x$ and $y$ stay weakly connected.

If $(x,y) \in E^{Tri}_e(t)$ then there can be the following cases:

\begin{itemize}
  \item $y \in x.P^+ \cup x.S^+$ at time $t$, then either $y$ is delegated to $x.pred^+_1$ (resp. $x.succ^+_1$) by a build-triangle message send to $x.pred^+_1$ (resp. $x.succ^+_1$) and $x$ and $y$ are connected in $G^{Tri}(t+1)$ by $(x,x.pred^+_1) \in E^{Tri}_e(t+1)$ and $(x.pred^+_1,y) \in E^{Tri}_i(t+1)$, or $y$ is stored in $x.P^+$ and $(x,y)\in E^{Tri}_e(t+1)$.
  \item $y \in x.S^- \cup x.P^-$ at time $t$, then either $y$ is delegated to $v^{-}(y)=\argmax \{h(v) : v \in x.S^* \wedge h(v)< h(y) \}$ (resp. $v^{+}(y)$) or $y$ is also stored in $x.S^-$ (resp. $x.P^-$) at time $t+1$. By the same arguments as above $x$ and $y$ are connected in $G^{Tri}(t+1)$.
\end{itemize}

If $(x,y) \in E^{Tri}_i(t)$ then $m=(build-triangle,y)\in Ch_x$. If $x$ processes $m$, then either $x$ is stored in $x.P^+, x.S^+, x.S^-,x.P^-$ and $(x,y)\in E^{Tri}_e(t+1)$ or $y$ is delegated to $v^{-}(y)=\argmax \{h(v): v \in x.S^* \wedge h(v)<h(y) \}$ (resp. $v^{+}(y)$) and $x$ and $y$ are connected in $G^{Tri}(t+1)$.
\end{proof}

\begin{lemma}\label{lem:TriConnectivity2}
If $G$ is weakly connected then eventually $G^{Tri}$ will be weakly connected.
\end{lemma}

\begin{proof}
Again we consider every edge $(x,y)$ in $G$ and show that eventually $x$ and $y$ will be connected in $G^{Tri}$. Note that we already showed in Lemma \ref{lem:TriConnectivity}, that nodes connected in $G^{Tri}$ stay connected in $G^{Tri}$. Therefore we only have to consider those edges in $E-E^{Tri}$.

If $(x,y)\in E_e(t) - E^{Tri}_e(t)$ there can be the following case:

$y \in x.DS[i].ref$ at time $t$ and $x$ has received an interval-update message with a new reference for the data in $x.DS[i]$ or the data is deleted. Then in both cases $x$ delegates $y$ to $v^{-}(y)$ (resp. $v^{+}(y)$) and $x$ and $y$ are connected in $G^{Tri}$. If $x$ does not delegate $y$, then $x$ eventually sends an check-interval message to $y$. Then either $x \in y.P^+ \cup y.S^+$ or $y$ delegates $x$ and $x$ and $y$ are weakly connected in $G^{Tri}$. If $y$ is stored in an incoming message in $Ch_x$ then there can be the following cases:

\begin{itemize}
  \item $y$ is in a list in a list-update message. Then either $y$ is stored in a new list $x.P^+,x.S^+,x.S^-,x.P^-$ or $y$ is delegated to $v^{-}(y)$ (resp. $v^{+}(y)$) and $x$ and $y$ are weakly connected in $G^{Tri}$.
  \item $y$ is a reference in an interval-update message, then either $y$ is stored as a new reference for some data or if there is no corresponding data $y$ is delegated to $v^{-}(y)$ (resp. $v^{+}(y)$) and $x$ and $y$ eventually are weakly connected in $G^{Tri}$.
  \item $y$ is a reference in an store-data message, then $y$ is stored as a new reference in $x.DS$. And as already shown $x$ and $y$ are eventually weakly connected.
\end{itemize}
\end{proof}

Combining the Lemmas ~\ref{lem:GConnectivity}, ~\ref{lem:TriConnectivity} and ~\ref{lem:TriConnectivity2} leads to the following theorem:
\begin{theorem}\label{theo:phase1}
If $G$ is weakly connected at time $t$, then for some time $t'>t$ $G^{Tri}$ will be weakly connected at every time $t''>t'$.
\end{theorem}

\subsubsection{Phase 2: Linearization}
For the rest of the analysis we assume that all variables of each node are valid according to definition ~\ref{def:valid}, i.e. we assume that each node has performed one consistency check. In this phase we show that eventually all nodes form a sorted list. We therefore define another subtopology

$G^{List}=(V,E^{List}=E^{List}_e \cup E^{List}_i)$ with $E^{List}_e=\left\{(x,y)\in E_e:(y=x.P^*[1] \vee y=x.S^*[1])\right\}$ and $E^{List}_i=\left\{(x,y)\in E_i:m=(build-triangle,y)\in Ch_x\right\}$.

In the end $G^{List*}_e$ with $E^{List*}_e=\{(x,y)\in E_e:(y=x.P^*[1]\- \vee y=x.S^*[1]) \wedge x.P^*[1]= \argmax_{ v \in V} \left\{h(v) : h(v)<h(x) \right\}
\wedge x.S^*[1]=\-\argmin_{ v \in V} \left\{h(v) : h(v)>h(x)\right\}\}$ shall be formed.

\begin{theorem}\label{theo:phase2}
If $G^{Tri}$ is weakly connected eventually $G^{List}_e$ will be strongly connected and $G^{List}_e=G^{List*}_e$.
\end{theorem}
 
Before we can show the theorem we show some helping lemmas.

\begin{lemma}\label{lem:RConnectivity}
Eventually all nodes $x.S^*[i]$ and $x.S^*[i+1]$ (resp. $x.P^*[i]$ and $x.P^*[i+1]$) with $i<|x.S^*|$ will be connected and stay connected in every state after over nodes $w: h(x.S^*[i])<h(w)<h(x.S^*[i+1])$.
\end{lemma}

\begin{proof}
In the periodic action $x$ executes $Build-Triangle()$, in which $x$ introduces every pair of nodes $x.S^*[i],x.S^*[i+1]$ to each other. The connecting path only changes if w.l.o.g. $x.S^*[i]$ delegates $x.S^*[i+1]$, but then $x.S^*[i]$ can delegate $x.S^*[i+1]$ only to a node $v$ with $h(x.S^*[i])<h(v)<h(x.S^*[i+1])$. By using this argument inductively $x.S^*[i]$ and $x.S^*[i+1]$ stay connected in every state after over nodes $h(x.S^*[i])<h(w)<h(x.S^*[i+1])$.
\end{proof}

\begin{lemma}\label{lem:Forward}
If $(x,y) \in E^{Tri}$ and $(x,z) \in E^{Tri}$ and $h(x)<h(y)<h(z)$ (resp. $h(x)>h(y)>h(z)$) then eventually $(v,z)\in E^{Tri}$ with $h(x)<h(v)<h(z)$ (resp. $h(x)>h(v)>h(z)$) and $x$ and $z$ are connected over nodes $w : h(x)<h(w)\leq h(v)$.
\end{lemma}

\begin{proof}
If $x$ delegates $z$ to a node $v$ then obviously $h(x)<h(v)<h(z)$ and $x$ and $z$ are connected over $w=v$.
In case $z$ is not delegated $z$ is stored in $x.S^+$ or $x.S^-$ (resp. $x.P^+$ or $x.P^-$) or in a message $m=(build-triangle,z) \in Ch_x$. If $z$ is stored in $x.S^-$ and  $(x,y) \in E^{Tri}$ and $h(x)<h(y)<h(z)$ then either $y \in x.S^-$ or $m'=(build-triangle,y) \in Ch_x$. Eventually $m'$ is processed by $x$ and either $y$ is delegated, then there is another node $v'\in x.S^- : h(x)<h(v')<h(y)<h(z)$ or $y$ is stored in $x.S^-$. Thus eventually $z \in x.S^-$ and another node $v'\in x.S^-$ such that $h(x)<h(v')<h(z)$. From all such nodes $v'\in x.S^-$ such that $h(x)<h(v')<h(z)$ $x$ introduces $z$ to $v=\argmax \left\{h(w): w \in x.S^- \wedge h(w)<h(z) \right\}$ and $(v,z)\in E^{Tri}$ with $h(x)<h(v)<h(z)$ and by the same arguments as above $x$ and $z$ stay connected over nodes $h(x)<h(w)\leq h(v)$.
If $z \in x.S^+$ and $z= x.succ^+_1$ the same analysis as for $z\in x.S^-$ can be applied. If $z \in x.S^+$ and $z\neq x.succ^+_1$ eventually $x$ will receive the $v.S^+$ list of $v=x.succ^+_1$. If $z \in v.S^+$ then $(v,z)\in E^{Tri}$ with $h(x)<h(v)<h(z)$ and by the same arguments as above $x$ and $z$ stay connected over nodes $h(x)<h(w)<h(z)$. Otherwise $x$ sends a message $m=(build-triangle,z)$ to $v$ and again $(v,z)\in E^{Tri}$ with $h(x)<h(v)<h(z)$ and by the same arguments as above $x$ and $z$ stay connected over nodes $h(x)<h(w)<h(v)$.
If $m=(build-triangle,z) \in Ch_x$, then eventually $x$ processes $m$ and either stores $z$ in $x.S^+$ or $x.S^-$ and we can apply one of the cases above or $z$ is delegated.
\end{proof}

\begin{lemma}\label{lem:Mirror}
If $(x,z) \in E^{Tri}$ with $h(x)<h(z)$ then eventually $(z,y)\in E^{Tri}$ with $h(x)<h(y)<h(z)$ and $x$ and $y$ are connected over nodes $w: h(x)<h(w)<h(y)$.
\end{lemma}

\begin{proof}
If $(x,z) \in E^{Tri}$ and $(x,y) \in E^{Tri}$ and $h(x)<h(y)<h(z)$ we can apply Lemma ~\ref{lem:Forward} and eventually $(v,z)\in E^{Tri}$ with $h(x)<h(v)<h(z)$ (resp. $h(x)>h(v)>h(z)$) and $v$ and $z$ are connected over nodes $w:$ $h(x)<h(w)<h(z)$ in every state after. Now if $(v,y')\in E^{Tri}$ with $h(v)<h(y')<h(z)$ we might again apply the lemma. Obviously we only can apply Lemma ~\ref{lem:Forward} a finite number of times until there is a node $v'$ such that $(v',z)\in E^{Tri}$ and there is no $(v',y')\in E^{Tri}$ with $h(v')<h(y')<h(z)$ and $x$ and $z$ are connected over nodes $h(x)<h(w)<h(v')$. Then either $z=v'.S^*[1]$ or $m=(build-triangle,z) \in Ch_{v'}$. If $z=v'.S^*[1]$ then eventually $v'$ will introduce itself to $z$ by a message $m'=(build-triangle,v')$, then $(z,v')\in E^{Tri}$ with $h(x)<h(v')<h(z)$ and $x$ and $v'$ are connected over nodes $w:$ $h(x)<h(w)<h(v')$.
Otherwise as soon as $v'$ processes $m'$ $v'.S^*[1]$ is set to $z$ and the same arguments as in the first case hold.
\end{proof}

Before we prove the theorem we introduce some additional definitions.

\begin{definition}
In the directed graph we define an \emph{undirected path} $p$ as a sequence of edges $(v_0,v_1),(v_1,v_2),$ \newline $\cdots , (v_{k-1},v_k$), such that $\forall i \in \left\{1,\cdots ,k\right\}:(v_i,v_{i-1})\in E^{Tri} \vee (v_{i-1},v_i)\in E^{Tri}$. Let $u_{min}=\argmin \left\{h(v): v \in p\right\}$ and $u_{max}=\argmax \left\{h(v): v \in p\right\}$ then the \emph{range of a path} $range(p)$ is given by $range(p)=u_{max}-u_{min}$.
\end{definition}

Now we are ready to prove Theorem ~\ref{theo:phase2}.

\begin{proof}
Let $x$ and $y$ be a pair of nodes connected in $G^{List*}_e$; i.e. w.o.l.g. $x=\argmax \left\{h(v) : v \in V \wedge v<y\right\}$ and $y=\argmin \left\{h(v) : v \in V \wedge h(v)>h(x) \right\}$. Then as $G^{Tri}$ is weakly connected there is an undirected path connecting $x$ and $y$. Let $p(t)$ be such a path at time $t$. We show that there is a path $p(t')$ with $t'>t$ that connects $x$ and $y$ weakly such that $range(p(t))>range(p(t'))$. Let $u_{min}$ and $u_{max}$ be the smallest and greatest node on the path $p(t)$ that limit the range of $p(t)$. Then $u_{min}$ is connected to nodes $w_1$ and $w_2$. If $(w_1,u_{min})\in E^{Tri}$ and $(u_{min},w_1)\notin E^{Tri}$. Then according to Lemma ~\ref{lem:Mirror} eventually $(u_{min},v_1)\in E^{Tri}$ and $w_1$ and $v_1$ are connected over nodes $w'$ such that $h(v_1)<h(w')<h(w_1)$. The same holds for $w_2$. Thus eventually $(u_{min},v_1)\in E^{Tri}$ and $(u_{min},v_2)\in E^{Tri}$ and $w_1$ and $v_1$ are connected over nodes $w'$ with $h(v_1)<h(w')<h(w_1)$ and $w_2$ and $v_2$ are connected over nodes $w'': h(v_2)<h(w'')<h(w_2)$. Then either $v_1=v_2$ and we can construct another path connecting $x$ and $y$ over $w_1$ and $w_2$ with $u'_{min}=v_1=v_2$, and $h(u'_{min})>h(u_{min})$, otherwise $h(v_1)<h(v_2)$ or $h(v_2)<h(v_1)$. W.l.o.g. we assume $h(v_1)<h(v_2)$. Then according to Lemma ~\ref{lem:Forward} eventually $(v'_2,v_2)\in E^{Tri}$ and $v'_2$ and $u_{min}$ are connected over nodes $w'' : h(u_{min})<h(w'')<h(v'_2)$. Either $v_1=u_{min}.S^*[1]$ or also according to Lemma ~\ref{lem:Forward}  $(v'_1,v_1)\in E^{Tri}$ and $v'_1$ and $u_{min}$ are connected over nodes $w'': h(u_{min})<h(w'')'<h(v'_1)$. Note that according to the proof of Lemma ~\ref{lem:Forward} $v'_1$ and $v'_2$ have to be in $u_{min}.S^*$ at the time the edge $(v'_2,v_2)\in E^{Tri}$ resp. $(v'_1,v_1)\in E^{Tri}$ is created. Then according to Lemma ~\ref{lem:RConnectivity} $v'_1$ and  $v'_2$ are also connected to $u_{min}.S^*[1]$. Thus again we can construct another path connecting $x$ and $y$ over $w_1$ and $w_2$ with $u'_{min}=u_{min}.S^*[1]\wedge h(u_{min}.S^*[1])>h(u_{min})$. The same arguments can be used symmetrically to show that $u_{max}$ can be decreased. Thus eventually a connecting path can be found with a strict smaller range and by applying these arguments a finite number of times $(x,y)\in E^{Tri}$ and $(y,x) \in E^{Tri}$. Then if $y \in x.S^*$ $x.S^*[1]=y$ otherwise $m=(build-triangle,y) \in Ch_x$ will eventually be processed and $x.S^*[1]$ is set to $y$. By the same arguments eventually $x=y.P^*[1]$. As this holds for every pair $x$, $y$ in $G^{List*}_e$, eventually $G^{List}_e=G^{List*}_e$.
\end{proof}

\subsubsection{Phase 3: From the sorted list to the CONE-network}

In this section we show that once the network has stabilized into a sorted list, it eventually also stabilizes into the legal cone-network, that means,
each node $u$ maintains a correct set of neighbors, so the lists $u.P^+,u.S^+,u.P^-,u.S^-$ maintain the correct nodes, so for example the list $u.P^+$ maintains
the nodes in $P^+(u)$. For all the following lemmas and theorems in this section we assume $G^{List}_e=G^{List*}_e$.

We will first do the proof for the sets $u.S^-$ and $u.S^+$.
The following lemma will be helpful.

\begin{lemma}\label{induct1}
If every node at the right (with larger id) of a node $u$ knows its correct closest larger right node $succ^+_1(u)$ (stored in $u.succ^+_1$), then for all nodes $x$ which are in the correct right internal neighborhood  of $u$, $S^-(u)$, it holds that $u$ will eventually learn $x$ (and store it in $u.S^-$).
\end{lemma}

\begin{proof}
We will prove it by induction over $x$ (in ascending order of their $h(x)$ values).

\textbf{Induction basis:} $x$ is the direct right neighbor of $u$.  
In that case $u$ already knows $x$, since we assumed the presence of the sorted list, and the statement holds.

\textbf{Inductive step:} If $u$ knows the next node to the left of $x$ (let this be $y$) which is in the right internal neighborhood of $u$, then $u$ eventually learns $x$.
In this case, $x$ is the closest larger right node of $y$. That is because $x$ must be larger (in terms of capacity) than $y$, since else $x$ would not
be in $S^-(u)$. By hypothesis, $y$ knows about $x$ (so $y.succ^+_1=x$). So, when $y$ conducts its periodic $BuildTriangle$ call, it will introduce $u$ and $x$ to each other (as they are $y.pred^+_1$ and $y.succ^+_1$ respectively)
and $u$ will learn about $x$.
\end{proof}

Now we show that eventually all nodes learn their correct right larger-node lists.

\begin{lemma}\label{S^+}
Once the list has been established and $G^{List}_e=G^{List*}_e$, then eventually every node $u$ learns its correct right larger-node list $S^+(u)$ (and stores it in $u.S^+$).
\end{lemma}

\begin{proof}
We will prove this by induction over the nodes $u$ (in descending order of their $h(u)$ values).

\textbf{Induction basis:} $u$ does not have any closest larger right node $succ^+_1(u)$. The statement is obliviously true.

\textbf{Inductive step:} If every node at the right of $u$ in the list, knows its correct right larger-node list, then eventually $u$ will learn its correct right larger-node list $S^+(u)$.

From the induction hypothesis, every node at right of $u$ knows its correct right larger-node list (so also its correct closest larger right node), so according to Lemma ~\ref{induct1}, $u$ will eventually learn its correct right internal neighborhood (and store it in $u.S^-$). Let $v$ the node being the most right one in $u.S^-$. It is obvious that the  closest larger right node of $v$, $succ^+_1(v)$, is also the closest larger right node of $v$, since otherwise $succ^+_1(v)$ would also be in $u.S^-$. By the inductive hypothesis, $v$ knows this node, and will introduce it to $u$ by its periodic $BuildTriangle$ call. So, once $u$ learns $succ^+_1(v)$ (and as a consequence $succ^+_1(v)$ learns $u$ after $u$'s periodic $BuildTriangle$ call, $succ^+_1(v)$ will also send to $u$ its right larger-node list ($u.S^+$), through its periodic $list-update$ message, which (together with $succ^+_1(v)$) is the correct right larger-node list of $u$.
\end{proof}

\begin{lemma}\label{S^-}
If $G^{List}_e=G^{List*}_e$, then eventually every node $u$ learns its correct right internal neighborhood $S^-(u)$ (and stores it in $u.S^-$).
\end{lemma}

\begin{proof}
By Lemma ~\ref{S^+}, there is a point where every node knows its right larger-node list $S^+(u)$ .
This is the hypothesis of Lemma ~\ref{induct1} for all nodes $u$, so by using this lemma for every node $u$ we derive the proof.
\end{proof}

\begin{theorem}\label{theo:phase3}
If $G^{List}_e=G^{List*}_e$ then eventually, every node $u$ learns its correct internal neighborhood $S^-(u), P^-(u)$ , as well as its correct larger-node lists $S^+(u),P^+(u)$.
\end{theorem}

\begin{proof}
We already showed that for the right part of the neighborhood (for $S^+(u)$ and $S^-(u)$) by lemmas ~\ref{S^+} and ~\ref{S^-}.
By symmetry (i.e. by using symmetric proofs for the left part) it also holds for $P^+(u)$ and $P^-(u)$.
\end{proof}


Combining Theorem ~\ref{theo:phase1}, Theorem ~\ref{theo:phase3} and Theorem ~\ref{theo:phase3} we can show that Theorem ~\ref{theo:convergence} holds, and by our protocol each weakly connected network converges to a CONE network.

\subsubsection{Closure and Correctness of the data structure}

We showed that from any initial state we eventually reach a state in which the network forms a correct CONE network. We now need to show that in this state the explicit edges remain stable and also that each node stores the data it is responsible for.

\begin{theorem}\label{theo:closure}
If $G_e=G^{CONE}$ at time $t$ then for $t'>t$ also $G_e=G^{CONE}$.
\end{theorem}

\begin{proof}
The graph $G_e=(V,E_e)$ only changes if the explicit edge set is changed. So if we assume that $G_e=G^{CONE}$ at time $t$ and for $t'>t$ also $G_e\neq G^{CONE}$ then we added or deleted at least one explicit edge. Let $(u,v)\in E_e$ at time $t$. W.l.o.g. we assume $h(u)<h(v)$. Either $v \in S^+(u)$ or $v\in S^-(u)$. In both cases the edge is only deleted if $u$ knows a node $w$ with $h(u)<h(w)<h(v)$ and $c(w)>c(v)$ as following from Theorem ~\ref{theo:phase3} all internal neighborhoods are correct in $G^{CONE}$ there can not be such a node $w$. By the same argument also no new edges are created. Thus $G_e=G^{CONE}$ at time $t'$.
\end{proof}

So far we have shown that by our protocol eventually a CONE network is formed. It remains to show that also by our protocol eventually each node stored the data it is responsible for.

\begin{theorem}\label{theo:datastructure}
If $G_e=G^{CONE}$ eventually each node stores exactly the data it is responsible for.
\end{theorem}

\begin{proof}
According to Theorem ~\ref{theo:responsibility} each node knows which node is responsible for parts of the interval it supervises. In our described algorithm each node $u$ checks whether it is responsible for the data it currently stores by sending a message to the node $v$ that $u$ assumes to be supervising the corresponding interval. If $v$ is supervising the interval and $u$ is responsible for the data, then $u$ simply keeps the data. If $v$ is not supervising the data or $u$ is not responsible for the data then $v$ sends a reference to $u$ with the id of anode that $v$ assumes to be supervising the interval. Then $u$ forwards the data to the new reference and does not store the data. By forwarding the data by Greedy Routing it eventually reaches a node supervising the corresponding interval, this node then tells the responsible node to store the data. Thus eventually all data is stored by nodes that are responsible for the data.
\end{proof}

\section{External Dynamics}\label{ops}

Concerning the network operations in the network, i.e. the joining of a new node, the leaving of a node and the capacity change of a node, we show the following:

\begin{theorem}\label{theo:dynam}
In case a node $u$ joins a stable CONE network, or a node $u$ leaves a stable CONE network 
or a node $u$ in a stable CONE network changes its capacity, we show that in any of these three cases $\mathcal O(\log^2 n)$ structural changes in the explicit edge set are necessary to reach the new stable state.
\end{theorem}

We show the statement by considering the 3 cases separately.

\subsection{Joining of a node}

When a new node $u$ enters the network, it does so by maintaining a connection to another node $v$, which is already in the network. $u$ is forwarded due to the $BuildTriangle$ and $ListUpdate$ rules in the network until it reaches its right position, as it takes part in the linearization procedure.

\begin{theorem}
If a node $u$ joins a stable CONE network $\mathcal O(\log^2 n)$ structural changes in the explicit edge set are necessary to reach the new stable state.
\end{theorem}

\begin{proof}
We show that there is at most a constant number of temporary edges, i.e. edges that are not in the stable state.
$u$ stores $v$ in its internal variables $u.P^*$ or $u.S^*$ as $v$ is the only node $u$ knows. In the periodic BuildTriangle() $u$ sends a message to $v$ containing its own id creating an implicit edge $(v,u)$. Now there can be two cases: Either $u$ is in $v$' lists $v.P^+$, $v.S^+$, $v.P^-$, $v.S^-$ in a stable state then $v$ stores $u$'s identifier or $u$ is not stored and delegated to another node $w$ creating the implicit edge $(w,u)$. Thus only explicit edges pointing to $u$ are created that are in the stable state and only the explicit edge $(u,v)$ is temporary. So far we have shown that according to Theorem ~\ref{theo:phase3} and ~\ref{theo:degree} at most $\mathcal O(\log n)$ edges are created w.h.p. that point to $u$ or start at $u$. But by the join of $u$ to the network also edges that have been in the stable state not longer exist in the new stable state. E.g. let $w\in x.S^+$ and $h(x)<h(u)<h(w)$ and $c(u)>c(w)$ then $w$ is not longer stored in $x.S^+$ as soon as $u$ is integrated in the network. According to ~\ref{theo:degree} there is at most $\mathcal O(\log n)$ w.h.p. such nodes $x$, as each node $x$ has to store $u$ in its lists, and also at most $\mathcal O(\log n)$ w.h.p. nodes $w$, as $u$ has to store each $w$ in its lists. Therefore there are at most $\mathcal O(\log^2 n)$ edges that have to be deleted.
\end{proof}

\subsection{Leaving of a node}

Once a node decides it wants to leave the network, it disconnects itself from its neighbors. In the case it is the node with the greatest capacity, it introduces its two direct neighbors to each other before doing so. In that way, connectivity is still guaranteed (at least in the stable state).

\begin{algorithm}
\begin{algorithmic}
\Function{Leave}{}
\ForAll{$data \in u.DS$}
\State  send m=forward-data(data,false) to $u.S^-[1]$
\State delegate the references in $u.DS$
\EndFor
\If {$u.P^+= \emptyset \wedge u.S^+=\emptyset$}
\State send m=(buildtriangle, $u.P^-[1]$) to $u.S^-[1]$
\EndIf
\State delete all connections, leave network
\EndFunction
\end{algorithmic}
\end{algorithm}

After the leaving, the network must stabilize again.
This means that $u.S^-[1]$ and $u.P^-[1]$ must connect to each other. Lets consider $u.P^-[1]$. Since it won't have a direct right neighbor after the leaving of $u$, the linearization process will take place again until $u.P^-[1]$ learns $u.S^-[1]$.

\begin{theorem}
If a node $u$ leaves a stable CONE network $\mathcal O(\log^2 n)$ structural changes in the explicit edge set are necessary to reach the new stable state.
\end{theorem}

\begin{proof}
The proof is analogous to the proof in the case of a joining node. Obviously according to ~\ref{theo:degree} w.h.p. $\mathcal O(\log n)$ edges are deleted that start at or point to the leaving node $u$. By deleting $u$ further edges have to be created. E.g. let $w\in u.S^-$ and $u \in x.S^+$ and $c(u)>c(w)$ then $w$ might now be stored in $x.S^+$ or $x.S^-$ and the edge $(x,w)$ has to be created. Again according to ~\ref{theo:degree} there are w.h.p. at most $\mathcal O(\log n)$ such nodes $x$ and  $\mathcal O(\log n)$ such nodes $w$, thus in total at most $\mathcal O(\log^2 n)$ edges have to be created.
\end{proof}


\subsection{Capacity Change}
If the capacity of a single node $u$ in a stable CONE network decreases we can apply the same arguments as for the leaving of a node, as some nodes might now be responsible for intervals that $u$ was responsible for. Additionally $u$ might have to delete some ids in $u.S^- \cup u.P^-$ and add ids in $u.S^+ \cup u.P^+$. If a node increases its capacity we can apply the same arguments as for the joining of a node, as some nodes might not longer be responsible for intervals that $u$ is now responsible for. Additionally $u$ might have to add some ids in $u.S^- \cup u.P^-$ and delete ids in $u.S^+ \cup u.P^+$. Thus the following theorem follows.

\begin{theorem}
If a node $u$ in stable CONE network changes its capacity $\mathcal O(\log^2 n)$ structural changes in the explicit edge set are necessary to reach the new stable state.
\end{theorem}


\section{Conclusion and Future Work}

We studied the problem of a self-stabilizing and heterogeneous overlay network and gave an algorithm of solving that problem, and by doing this we used an efficient network structure. We proved the correctness of our protocol, also concerning the functionality of the operations done in the network, data operations and node operations. This is the first attempt to present a self-stabilizing method for a heterogeneous overlay network and it works efficiently regarding the information stored in the hosts. Furthermore our solution provides a low degree, fair load balancing and polylogarithmic updates cost in case of joining or leaving nodes.
In the future we will try to also examine heterogeneous networks in the two-dimensional space and consider heterogeneity in other aspects than only the capacity, e.g. bandwidth, reliability or heterogeneity of the data elements.





%

\end{document}